\documentclass[a4paper,UKenglish]{lipics-v2019}
\hideLIPIcs
\usepackage[utf8]{inputenc}
\usepackage{amssymb} 
\usepackage{amsmath}
\usepackage{amsfonts}
\usepackage{amsthm}
\usepackage{microtype}
\usepackage{url}
\usepackage{hyperref}
\usepackage{textcomp}
\usepackage{float}
\usepackage{enumitem}
\usepackage{tikz}
\usepackage{xspace}
\usepackage{comment}
\usepackage{bm}
\usepackage{todonotes}
\usetikzlibrary{backgrounds,plothandlers,plotmarks,matrix,calc}
\usepackage{nicematrix}

\def\cqedsymbol{\ifmmode$\lrcorner$\else{\unskip\nobreak\hfil
\penalty50\hskip1em\null\nobreak\hfil$\lrcorner$
\parfillskip=0pt\finalhyphendemerits=0\endgraf}\fi}
\newcommand{\cqed}{\renewcommand{\qed}{\cqedsymbol}}

\newcommand{\ZZ}{{\ensuremath{\mathbb{Z}}}}
\newcommand{\NN}{{\ensuremath{\mathbb{N}}}}
\newcommand{\Oh}{\ensuremath{\mathcal{O}}}
\newcommand{\Ohstar}{\ensuremath{\Oh^\star}}

\newcommand{\ve}[1]{\ensuremath{\mathbf{#1}}}
\newcommand{\vea}{\ensuremath{\mathbf{a}}}
\newcommand{\veb}{\ensuremath{\mathbf{b}}}
\newcommand{\vecc}{\ensuremath{\mathbf{c}}}
\newcommand{\ved}{\ensuremath{\mathbf{d}}}
\newcommand{\veg}{\ensuremath{\mathbf{g}}}
\newcommand{\vel}{\ensuremath{\mathbf{l}}}
\newcommand{\veu}{\ensuremath{\mathbf{u}}}
\newcommand{\vev}{\ensuremath{\mathbf{v}}}
\newcommand{\vew}{\ensuremath{\mathbf{w}}}

\newcommand{\velambda}{\ensuremath{\bm{\lambda}}}
\newcommand{\vex}{\ensuremath{\mathbf{x}}}

\renewcommand{\geq}{\geqslant}
\renewcommand{\leq}{\leqslant}

\renewcommand{\setminus}{-}

\newcommand{\td}{\mathrm{td}}
\newcommand{\dualtd}{\mathrm{td}_D}
\newcommand{\primtd}{\mathrm{td}_P}
\newcommand{\dualtw}{\mathrm{tw}_D}

\newcommand{\Graver}{\mathcal{G}}

\newcommand{\cfleq}{\sqsubseteq}

\newcommand{\cflt}{\sqsubset}

\title{Tight complexity lower bounds for integer linear programming with few constraints}

\titlerunning{Tight lower bounds for ILP with few constraints}

\author{Du\v{s}an Knop}{Algorithmics and Computational Complexity, Faculty~IV, TU Berlin \emph{and} \\
        Department of Theoretical Computer Science, Faculty of Information Technology,\\ Czech Technical University in Prague, Prague, Czech Republic
}{dusan.knop@fit.cvut.cz}{0000-0003-2588-5709}{Supported by DFG, project ``MaMu'', NI 369/19.}

\author{Micha\l{} Pilipczuk}{Institute of Informatics, University of Warsaw, Warsaw, Poland}{michal.pilipczuk@mimuw.edu.pl}{0000-0001-7891-1988}{}

\author{Marcin Wrochna}{Institute of Informatics, University of Warsaw, Warsaw, Poland \emph{and} University of Oxford, UK}{m.wrochna@mimuw.edu.pl}{0000-0001-9346-2172}{Supported by the Foundation for Polish Science (FNP) via the START stipend programme.}

\authorrunning{D. Knop, Mi. Pilipczuk, M. Wrochna}

\Copyright{Du\v{s}an Knop, Micha\l{} Pilipczuk, Marcin Wrochna}



\category{}

\relatedversion{}

\supplement{}

\funding{This work is a part of projects CUTACOMBS, PowAlgDO (M. Wrochna) and TOTAL (M. Pilipczuk) that have received funding from the European Research Council (ERC) under the European Union’s Horizon 2020 research and innovation programme (grant agreements No. 714704, No 714532, and No. 677651, respectively); the authors acknowledge the support of the OP VVV MEYS funded project CZ.02.1.01/0.0/0.0/16\_019/0000765 "Research Center for Informatics"
}

\acknowledgements{}

\EventEditors{Rolf Niedermeier and Christophe Paul}
\EventNoEds{2}
\EventLongTitle{36th International Symposium on Theoretical Aspects of Computer Science (STACS 2019)}
\EventShortTitle{STACS 2019}
\EventAcronym{STACS}
\EventYear{2019}
\EventDate{March 13--16, 2019}
\EventLocation{Berlin, Germany}
\EventLogo{}
\SeriesVolume{126}
\ArticleNo{41}
\nolinenumbers 

\begin{document}

\maketitle
\begin{abstract}
We consider the standard {\textsc{ILP Feasibility}} problem: given an integer linear program of the form $\{A\vex = \veb,\, \vex\geq 0\}$, 
where $A$ is an integer matrix with $k$ rows and $\ell$ columns, $\vex$ is a vector of $\ell$ variables, and $\veb$ is a vector of $k$ integers,
we ask whether there exists $\vex\in \NN^\ell$ that satisfies $A\vex = \veb$.
Each row of $A$ specifies one linear {\em{constraint}} on $\vex$; our goal is to study the complexity of {\textsc{ILP Feasibility}} when both 
$k$, the number of constraints, and $\|A\|_\infty$, the largest absolute value of an entry in $A$, are small.

Papadimitriou~\cite{Papadimitriou81} was the first to give a fixed-parameter algorithm for {\textsc{ILP Feasibility}} under 
parameterization by the number of constraints that runs in time $\left((\|A\|_\infty+\|b\|_\infty) \cdot k\right)^{\Oh(k^2)}$.
This was very recently improved by Eisenbrand and Weismantel~\cite{EisenbrandW18}, who used the Steinitz lemma to design an algorithm with running time $(k\|A\|_\infty)^{\Oh(k)}\cdot \|\veb\|_\infty^2$, which was
subsequently improved by Jansen and Rohwedder~\cite{JansenR18} to $\Oh(k\|A\|_\infty)^{k}\cdot \log \|\veb\|_\infty$.
We prove that for $\{0,1\}$-matrices $A$, the running time of the algorithm of Eisenbrand and Weismantel is probably optimal:
an algorithm with running time $2^{o(k\log k)}\cdot (\ell+\|\veb\|_\infty)^{o(k)}$ would contradict the Exponential Time Hypothesis (ETH).
This improves previous non-tight lower bounds of Fomin et al.~\cite{FominPRS16}.

We then consider integer linear programs that may have many constraints, but they need to be structured in a ``shallow'' way.
Precisely, we consider the parameter {\em{dual treedepth}} of the matrix $A$, denoted $\dualtd(A)$, 
which is the treedepth of the graph over the rows of $A$, where two rows are adjacent if
in some column they simultaneously contain a non-zero entry.
It was recently shown by Kouteck\'y et al.~\cite{KouteckyLO18} that {\textsc{ILP Feasibility}} can be solved in time $\|A\|_\infty^{2^{\Oh(\dualtd(A))}}\cdot (k+\ell+\log \|\veb\|_\infty)^{\Oh(1)}$.
We present a streamlined proof of this fact and prove that, again, this running time is probably optimal: even assuming that all entries of $A$ and $\veb$ are in $\{-1,0,1\}$, 
the existence of an algorithm with running time $2^{2^{o(\dualtd(A))}}\cdot (k+\ell)^{\Oh(1)}$ would contradict the ETH.

\end{abstract}

\vskip -1cm
\begin{picture}(0,0)
\put(392,10)
{\hbox{\includegraphics[width=40px]{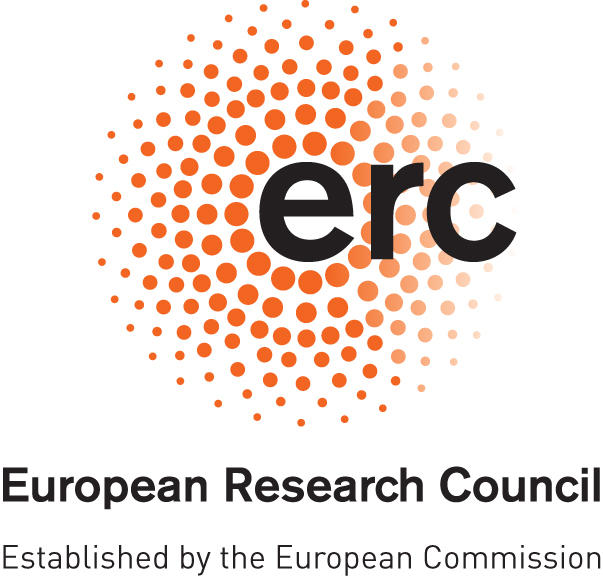}}}
\put(382,-50)
{\hbox{\includegraphics[width=60px]{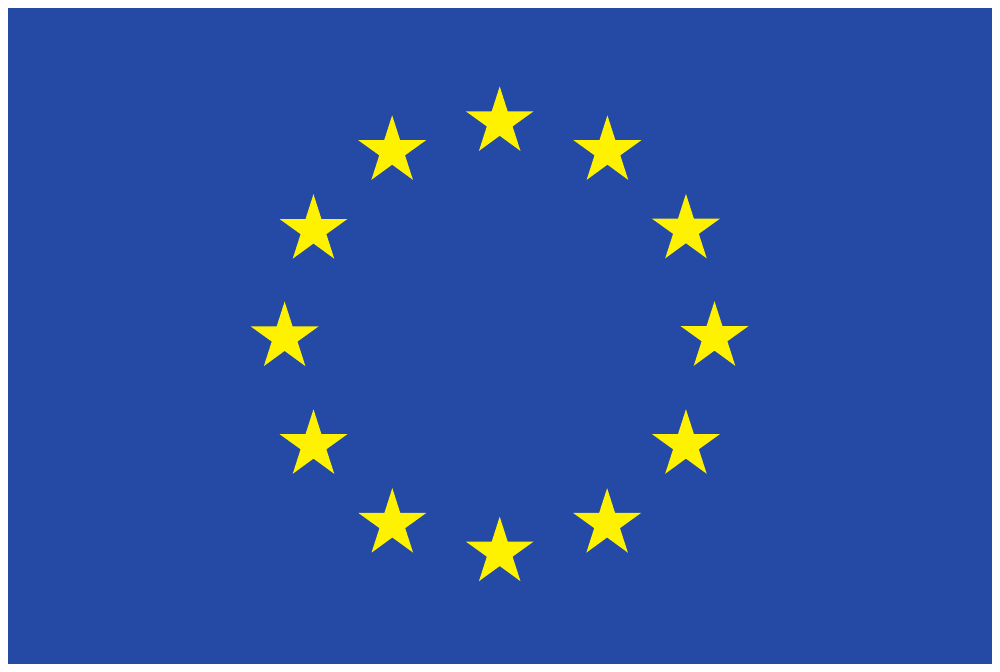}}}
\end{picture}


\section{Introduction}

Integer linear programming (ILP) is a powerful technique used in countless algorithmic results of theoretical importance,
as well as applied routinely in thousands of instances of practical computational problems every day.
Despite the problem being {\textsc{NP}}-hard in general, practical ILP solvers excel in solving real-life instances with thousands of variables and constraints.
This can be partly explained by applying a variety of subroutines, often based on heuristic approaches,
that identify and exploit structure in the input in order to apply the best suited algorithmic strategies.
A theoretical explanation of this phenomenon would of course be hard to formulate,
but one approach is to use the paradigm of {\em{parameterized complexity}}.
Namely, the idea is to design algorithms that perform efficiently when certain relevant structural parameters of the input have moderate values.

In this direction, probably the most significant is the classic result of Lenstra~\cite{lenstra}, who proved that {\textsc{ILP Optimization}} is {\em{fixed-parameter tractable}}
when parameterized by the number of variables $\ell$. That is, it can be solved in time $f(\ell)\cdot |I|^{\Oh(1)}$, where $f$ is some function and $|I|$ is the total bitsize of the input; we shall use the previous notation
throughout the whole manuscript.
Subsequent work in this direction~\cite{frank-tardos,kannan} improved the dependence of the running time on $\ell$ to $f(\ell)\leq 2^{\Oh(\ell \log \ell)}$.

In this work we turn to a different structural aspect and study ILPs that have few {\em{constraints}}, as opposed to few {\em{variables}} as in the setting considered by Lenstra.
Formally, we consider the parameterization by $k$, the number of {\em{constraints}} (rows of the input matrix $A$), and $\|A\|_\infty$, the maximum absolute value over all entries in $A$.
The situation when the number of constraints is significantly smaller than the number of variables appears naturally in many relevant settings.
For instance, to encode {\textsc{Subset Sum}} as an instance of {\textsc{ILP Feasibility}} it suffices to introduce a $\{0,1\}$-variable $x_i$ for every input number $s_i$,
and then set only one constraint: $\sum_{i=1}^n s_i x_i = t$, where $t$ is the target value.
Note that the fact that {\textsc{Subset Sum}} is {\textsc{NP}}-hard for the binary encoding of the input and polynomial-time solvable for the unary encoding, explains why $\|A\|_\infty$ is also a relevant parameter for
the complexity of the problem.
Integer linear programs with few constraints and many variables arise most often in the study of knapsack-like and scheduling problems via the concept of so-called {\em{configuration ILPs}},
in the context of approximation and parameterized algorithms.

\subparagraph*{Parameterization by the number of constraints.}
Probably the first to study the complexity of integer linear programming with few constraints was Papadimitriou~\cite{Papadimitriou81},
who already in 1981 observed the following. Consider an ILP of the standard form $\{A\vex = \veb,\, \vex\geq 0\}$, where $A$ is an integer matrix with $k$ rows (constraints) and $\ell$ columns (variables),
$\vex$ is a vector of integer variables, and $\veb$ is a vector of integers.
Papadimitriou proved that assuming such an ILP is feasible, it admits a solution with all variables bounded by $B=\ell \cdot ((\|A\|_\infty+\|\veb\|_\infty) \cdot k)^{2k+1}$,
which in turn can be found in time $\Oh((\ell B)^{k+1}\cdot |I|)$ using simple dynamic programming. Noting that by removing duplicate columns one can assume that $\ell\leq (2\|A\|_\infty+1)^k$, this
yields an algorithm with running time $((\|A\|_\infty + \|\veb\|_\infty)\cdot k)^{\Oh(k^2)}$. The approach can be lifted to give an algorithm with a similar running time bound
also for the {\textsc{ILP Optimization}} problem, where instead of finding any feasible solution $\vex$, we look for one that maximizes the value $\vew^{\intercal}\vex$ for a given optimization goal vector $\vew$.

The result of Papadimitriou was recently improved by Eisenbrand and Weismantel~\cite{EisenbrandW18}, who used the Steinitz Lemma to give an amazingly elegant algorithm
solving the {\textsc{ILP Optimization}} problem (and thus also the {\textsc{ILP Feasibility}} problem) for a given instance $\{\max \vew^{\intercal} \vex\colon A\vex = \veb,\, \vex\geq 0\}$
with $k$ constraints in time $(k\|A\|_\infty)^{\Oh(k)}\cdot \|b\|^2_\infty$.
This running time has been subsequently refined by Jansen and Rohwedder~\cite{JansenR18} to $\Oh(k\|A\|_\infty)^{2k}\cdot \log \|\veb\|_\infty$ in the case of {\textsc{ILP Optimization}}, and to
$\Oh(k\|A\|_\infty)^{k}\cdot \log \|\veb\|_\infty$ in the case of {\textsc{ILP Feasibility}}.\footnote{Throughout, $\log$ denotes the binary logarithm.}\looseness=-1

From the point of view of fine-grained parameterized complexity, this raises the question of whether the parametric factor $\Oh(k\|A\|_\infty)^{k}$ is the best possible.
Jansen and Rohwedder~\cite{JansenR18} studied this question under the assumption that $k$ is a fixed constant and $\|A\|_\infty$ is the relevant parameter.
They proved that assuming the Strong Exponential Time Hypothesis (SETH), for every fixed $k$ there is no algorithm with running time $(k\cdot (\|A\|_\infty+\|\veb\|_\infty))^{k-\delta}\cdot |I|^{\Oh(1)}$,
for any $\delta>0$.
Note that as $k$ is considered a fixed constant, this essentially shows that the degree of $\|A\|_\infty$ needs to be at least $k$,
but does not exclude algorithms with running time of the form $\|A\|_\infty^{\Oh(k)}\cdot |I|^{\Oh(1)}$, or $2^{\Oh(k)}\cdot |I|^{\Oh(1)}$ when all entries in the input matrix $A$ are in $\{-1,0,1\}$.
On the other hand, the algorithms of~\cite{EisenbrandW18,JansenR18} provide only an upper bound of $2^{\Oh(k\log k)}\cdot |I|^{\Oh(1)}$ in the latter setting.
As observed by Fomin et al.~\cite{FominPRS16}, a trivial encoding of {\textsc{3SAT}} as an ILP shows a lower bound of $2^{o(k)}\cdot |I|^{\Oh(1)}$ for instances with $A$ having entries only in $\{0,1\}$, $\veb$
having entries only in $\{0,1,2,3\}$, and $\ell =\Oh(k)$. This still leaves a significant gap between the $2^{o(k)}\cdot |I|^{\Oh(1)}$ lower bound and the $2^{\Oh(k\log k)}\cdot |I|^{\Oh(1)}$ upper bound.

\subparagraph*{Parameterization by the dual treedepth.}
A related, recent line of research concerns ILPs that may have many constraints, but these constraints need to be somehow organized in a structured, ``shallow'' way.
It started with a result of Hemmecke et al.~\cite{HemmeckeOR13}, who gave a fixed-parameter tractable algorithm for solving the so-called {\em{$n$-fold ILPs}}.
An $n$-fold ILP is an ILP where the constraint matrix is of the form $$A=\begin{pmatrix} B & B & \ldots & B\\ C & 0 & \cdots & 0 \\ 0 & C & \cdots & 0 \\ \vdots & \vdots & \ddots & \vdots \\ 0 & 0 & \cdots & C \end{pmatrix}$$
and the considered parameters are the dimensions of matrices $B$ and $C$, as well as $\|A\|_{\infty}$.
The running time obtained by Hemmecke et al. is $\|A\|_{\infty}^{\Oh(k^3)}\cdot |I|^{\Oh(1)}$ when all these dimensions are bounded by $k$.
See~\cite{HemmeckeOR13} and the recent improvements of Eisenbrand et al.~\cite{EisenbrandHKKLO19} for more refined running time bounds expressed in terms of particular dimensions.

The result of Hemmecke et al.~\cite{HemmeckeOR13} quickly led to multiple improvements
in the best known upper bounds for several parameterized problems, where the technique of configuration ILPs is applicable~\cite{KnopK18,KnopKM17b,KnopKM17a}.
Recently, the technique was also applied to improve the running times of several approximation schemes for scheduling problems~\cite{jansen-ptas}.
Chen and Marx~\cite{ChenM18} introduced a more general concept of {\em{tree-fold ILPs}}, where the ``star-like'' structure of an $n$-fold ILP is generalized to any bounded-depth rooted tree,
and they showed that it retains relevant fixed-parameter tractability results.
This idea was followed on by Eisenbrand et al.~\cite{EisenbrandHK18} and by Kouteck\'y et al.~\cite{KouteckyLO18}, whose further generalizations essentially
boil down to considering a structural parameter called the {\em{dual treedepth}} of the input matrix $A$.
This parameter, denoted $\dualtd(A)$, is the smallest number $h$ such that the rows of $A$ can be organized into a rooted forest of height $h$ with the following property:
whenever two rows have non-zero entries in the same column, one is the ancestor of the other in the forest.
As shown explicitly by Kouteck\'y et al.~\cite{KouteckyLO18} and somewhat implicitly by Eisenbrand et al.~\cite{EisenbrandHK18},
{\textsc{ILP Optimization}} can be solved in fixed-parameter time when parameterized by $\|A\|_{\infty}$ and $\dualtd(A)$.
For more detailed discussion of algorithmic implications and theory of block-structured integer programs we refer the reader to a recent survey~\cite{Chen19}.

\subparagraph*{Our results.}
For the parameterization by the number of constraints $k$, we close the above mentioned complexity gap by proving the following optimality result.

\begin{theorem}\label{thm:main}
Assuming ETH, there is no algorithm that would solve any {\textsc{ILP feasibility}} instance
$\{A \vex = \veb, \vex\geq 0\}$ with $A \in \{0,1\}^{k \times \ell}$, $\veb \in \NN^k$, and $\ell,\|\veb\|_\infty =\Oh(k \log k)$
in time~$2^{o(k\log k)}$.
\end{theorem}

This shows that the algorithms of~\cite{EisenbrandW18,JansenR18} have the essentially optimal running time of $2^{\Oh(k\log k)}\cdot |I|^{\Oh(1)}$ also in the regime where $\|A\|_\infty$ is a constant
and the number of constraints $k$ is the relevant parameter.
We can also reduce the coefficients in the target vector $\veb$ to constant at the cost of adding negative entries to $A$:

\begin{corollary}\label{cor:main}
	Assuming ETH, there is no algorithm that would solve any {\textsc{ILP feasibility}} instance
	$\{A \vex = \veb, \vex\geq 0\}$ with $A \in \{-1,0,1\}^{k \times \ell}$, $\veb \in \{0,1\}^k$ and $\ell=\Oh(k \log k)$
	in time~$2^{o(k\log k)}$.
\end{corollary}

The same cannot be done for non-negative matrices $A$, since in this case Papadimitriou's algorithm is even simpler and works in time $\|\veb\|_\infty^{\Oh(k)} \cdot |I|^{\Oh(1)}$.
The reduction in Theorem~\ref{thm:main} is hence simultaneously tight against this algorithm
(since for $\|\veb\|_\infty =\Oh(k \log k)$ the bound $\|\veb\|_\infty^{\Oh(k)}$ is $2^{\Oh(k \log k)}$).

The main ingredient of the proof of Theorem~\ref{thm:main} is a certain quaint combinatorial construction --- {\em{detecting matrices}} introduced by Lindstr\"om~\cite{Lindstrom65} --- that provides
a general way for compressing a system $A\vex = \veb$ with $k$ equalities and bounded targets $\|\veb\|_\infty \leq d$ into $\Oh(k/\log_d k)$ equalities (with unbounded targets).
Each new equality is a linear combination of the original ones; in fact, just taking $\Oh(k/\log_d k)$ sums of random subsets of the original equalities suffices, but we
also provide a deterministic construction taking $\Oh(dk/\log_d k)$ such subsets.
By composing such a compression procedure for $d=4$ with a standard reduction from {\textsc{(3,4)SAT}} ---
a variant of {\textsc{3SAT}} where every variable occurs at most $4$ times --- to {\textsc{ILP Feasibility}},
we obtain a reduction that given an instance of {\textsc{(3,4)SAT}} with $n$ variables and $m$ clauses, produces an equivalent instance of {\textsc{ILP Feasibility}} with $k=\Oh((n+m)/\log (n+m))$ constraints.
Since $2^{o(k\log k)}=2^{o(n+m)}$, we would obtain a $2^{o(n+m)}$-time algorithm for {\textsc{(3,4)SAT}}, which is known to contradict ETH.
We note that detecting matrices were recently used by two of the authors in the context of different lower bounds based on ETH~\cite{BonamyKPSW17}.

\bigskip

For the parameterization by the dual treedepth,
we first streamline the presentation of the approach of Kouteck\'y et al.~\cite{KouteckyLO18} and clarify that the parametric factor in the running time is doubly-exponential in the treedepth.
The key ingredient here is the upper bound on $\ell_1$-norms of the elements of the {\em{Graver basis}} of the input matrix $A$, expressed in terms of $\|A\|_\infty$ and $\dualtd(A)$.
Using standard textbook bounds for Graver bases and the recursive definition of treedepth, we prove that these $\ell_1$-norms can be bounded by $(2\|A\|_{\infty}+1)^{2^{\dualtd(A)}-1}$.
This, combined with the machinery developed by Kouteck\'y et al.~\cite{KouteckyLO18}, implies the following.

\begin{theorem}\label{thm:dualtd-upper-bound}
There is an algorithm that solves any given {\textsc{ILP Optimization}} instance $I=\{\max \vew^\intercal \vex \colon A\vex=\veb, \vel\leq \vex\leq \veu\}$ in time
$\|A\|_\infty^{2^{\Oh(\dualtd(A))}}\cdot |I|^{\Oh(1)}$.
\end{theorem}

We remark that the running time as outlined above also follows from a fine analysis of the reasoning presented in~\cite{KouteckyLO18},
but the intermediate step of using tree-fold ILPs in~\cite{KouteckyLO18} makes tracking parametric dependencies harder to follow.

We next show that, perhaps somewhat surprisingly, the running time provided by Theorem~\ref{thm:dualtd-upper-bound} is optimum. Namely, we have the following lower bound.

\begin{theorem}\label{thm:dualtd-lower-bound}
Assuming ETH, there is no algorithm that would solve any {\textsc{ILP Feasibility}} instance $I=\{A\vex = \veb, \vex \geq 0\}$,
where all entries of $A$ and $\veb$ are in $\{-1,0,1\}$, in time~$2^{2^{o(\dualtd(A))}}\cdot |I|^{\Oh(1)}$.
\end{theorem}

To prove Theorem~\ref{thm:dualtd-lower-bound} we reduce from the {\textsc{Subset Sum}} problem.
The key idea is that we are able to ``encode'' any positive integer $s$ using an ILP with dual treedepth $\Oh(\log \log s)$.
This lower bound has been recently generalized by Eisenbrand et al.~\cite{EisenbrandHKKLO19} to include a parameter they call topological height.

\section{Parameterization by the number of constraints}

\subsection{Detecting matrices}
Our main tool is the usage of so-called \emph{detecting matrices}, first studied by Lindstr\"{o}m~\cite{Lindstrom65}.
They can be explained via the following coin-weighing puzzle: given $m$ coins with weights in $\{0,1,\dots,d-1\}$, we want the deduce the weight of each coin with as few weighings as possible.
We have a spring scale, so in one weighing we can exactly determine the sum of weights of any subset of the coins.
While the naive strategy---weigh coins one by one---yields $m$ weighings, it is actually possible to find a solution using $\Oh(m/\log_d m)$ weighings.
This number is asymptotically optimal, as each weighing provides $\Theta(\log m)$ bits of information, so fewer weighings would not be enough to distinguish all $d^m$ possible weight functions.

Probably the easiest way to construct such a strategy is using the probabilistic method.
It turns out that querying $\Oh(m/\log_d m)$ random subsets of coins with high probability provides enough information to determine the weight of each coin.
This is because a random subset distinguishes any of the $\Oh(d^m \cdot d^m)$ non-equal pairs of weight functions with probability at least $\frac{1}{2}$,
but pairs of weight functions that are close to each other are few, while pairs of weight functions that are far from each other have a significantly better probability than $\frac{1}{2}$ of being distinguished.
Note that thus we construct a {\em{non-adaptive}} strategy: the subsets of coins to be weighed can be determined and fixed at the very start.
We refer the reader to e.g.~\cite[Corollary 2]{GrebinskiK00} for full details, and we remark that the last two authors recently used detecting matrices in the context of algorithmic lower bounds
for the {\textsc{Multicoloring}} problem~\cite{BonamyKPSW17}.

Viewing each tuple of coin weights as a vector $\ve{v}\in \{0,\dots,d-1\}^m$, each weighing returns the value $\ve{a}^\intercal \ve{v}$ for the characteristic vector $\ve{a} \in \{0,1\}^m$ of some subset of coins.
Thus $k$ weighings give the vectors of values $M \ve{v}$ for some $\{0,1\}$-matrix $M$ with $k$ rows and $m$ columns.
An equivalent formulation is then to ask for a $\{0,1\}$-matrix $M$ with $m$ columns, such that knowing the vector $M \ve{v}$ uniquely determines any $\ve{v}\in\{0,\dots,d-1\}^m$.
Such an $M$ is called a \emph{$d$-detecting matrix} and we seek to minimize the number of rows/weighings $k$ it can have.
Lindstr\"{o}m gave a deterministic construction and proved the bound on $k$ to be tight.
See also Bshouty~\cite{Bshouty09} for a more direct and general construction using Fourier analysis.

\begin{theorem}[\cite{Lindstrom65}]\label{thm:detectingBounded}
For all $d,m\in\NN$, there is a $\{0,1\}$-matrix $M$ with $m$ columns and $k\leq\frac{2m \log d}{\log m}(1+o(1))$ rows such that for any $\ve{u},\ve{v}\in \{0,\dots,d-1\}^m$,
if $M\ve{u} = M\ve{v}$ then~$\ve{u}=\ve{v}$.
Moreover, such matrix $M$ can be constructed in time polynomial in $dm$.
\end{theorem}

In other words, this allows us to check $m$ equalities between values in $\{0,\dots,d-1\}$ (i.e., corresponding coordinates of vectors $\ve{u}$ and $\ve{v}$)
using only $\Oh(m/\log_d m)$ comparisons of sums of certain subsets of these values (i.e., coordinates of vectors $M\ve{u}$ and $M\ve{v}$).
For an ILP instance $A \vex = \veb$ with  $\|\veb\|_\infty \leq d$ and $m$ constraints, we may use this idea to check the equality on each of the $m$ coordinates of $A \vex$
using only $\Oh(m/\log_d m)$ constraints.
Indeed, the intuition is that
if $M$ is a $d$-detecting matrix, then we can rewrite $A \vex = \veb$ as $M A\vex = M \veb$ and check the latter --- which involves $\Oh(m/\log_d m)$ $\{0,1\}$-combinations of the original constraints.

This is the core of our approach.
However, there is one subtle caveat: in order to claim that the assertions $A \vex = \veb$ and $M A\vex = M \veb$ are equivalent,
we would need to ensure that $\|A\ve{x}\|_\infty\leq d$ for an arbitrary vector $\ve{x} \in \NN^n$.
One solution is to use the fact that a uniformly random $\{0,1\}$-matrix has a stronger ``detecting'' property:
it will, with high probability, distinguish all vectors of low $\ell_1$-norm, as shown by Grebinski and Kucherov~\cite{GrebinskiK00}.

\begin{lemma}[\cite{GrebinskiK00}]\label{lem:detectingRandom}
For all $d,m\in\NN$, there exists a $\{0,1\}$-matrix $M$ with $m$ columns and $k\leq\frac{4m \log (d+1)}{\log m}(1+o(1))$ rows
such that for any $\ve{u},\ve{v}\in\NN^m$ satisfying $\|\ve{u}\|_1,\|\ve{v}\|_1 \leq dm$, if $M\ve{u} = M\ve{v}$ then $\ve{u}=\ve{v}$.
Moreover, such matrix $M$ can be computed in randomized polynomial time (in $dm$).
\end{lemma}

Note that in Lemma~\ref{lem:detectingRandom}, we do not actually have to assume bounds on one of the two vectors: it suffices to assume $\ve{u} \in \NN^m$ and $\|\ve{v}\|_1 \leq dm$,
because simply adding a single row full of ones to $M$ guarantees $\|\ve{u}\|_1=\|\ve{v}\|_1$.
Therefore as long as $A$ is non-negative and $\|\veb\|_\infty \leq d$, it suffices to check $MA\vex = M\veb$.
Unfortunately, to the best of our knowledge, no deterministic construction is known for Lemma~\ref{lem:detectingRandom}.
We remark that Bshouty gave a deterministic, but adaptive detecting strategy~\cite{Bshouty09}; that is, in terms of coin weighing, consecutive queries on coins may depend on results of previous weighings.

Instead, we show that a different, recursive construction by Cantor and Mills~\cite{Cantor66} for $2$-detecting matrices can be adapted so that no bounds (other than non-negativity)
are assumed for one of the vectors, while the other must have all coefficients in $\{0,1,\dots,d-1\}$.

\begin{lemma}\label{lem:detectingOne}
For all $d,m\in\NN$,
there exists a $\{0,1\}$-matrix $M$ with $m$ columns and \mbox{$k\leq\frac{md\log d}{\log m}(1+o(1))$} rows such that for any $\ve{u}\in\NN^m$ and $\ve{v}\in\{0,1,\dots,d-1\}^m$,
if $M\ve{u} = M\ve{v}$ then $\ve{u}=\ve{v}$. Moreover, such matrix $M$ can be computed in time polynomial in~$dm$.
\end{lemma}
\begin{proof}[Proof of Lemma~\ref{lem:detectingOne}]
Fix $d\geq 2$.
We construct inductively for each $i\in\NN$ a certain $\{0,1\}$-matrix $M_i$ with $k_i$ rows and $m_i$ columns
such that for any $\ve{u}\in\NN^{m_i}$ and $\ve{v}\in\{0,1,\dots,d-1\}^{m_i}$, if $M\ve{u} = M\ve{v}$ then $\ve{u}=\ve{v}$.
For $i=1$ we use the $d\times d$ identity matrix, that is, $k_1=m_1=d$.

\newcommand{\vx}[1]{{\vex^{(#1)}}}
\newcommand{\vxp}[1]{{\vex'^{(#1)}}}

Let $B=M_i$ be such the matrix for $i\geq 1$. We claim that the following matrix $M=M_{i+1}$ with $k_{i+1}=d\cdot k_{i}+d$ rows and $m_{i+1} = d\cdot m_{i}+k_{i}$ columns satisfies the same condition
\[
M = \begin{pmatrix}
  B   &   B   &\cdots&   B   &   I   \\
  B   &  J-B  &      &       &       \\
\vdots&       &\ddots&       &       \\
  B   &       &      &  J-B  &       \\ \hline
      &1\dots1&      &       &       \\
      &       &\ddots&       &       \\
      &       &      &1\dots1&       \\
      &       &      &       &1\dots1\\
\end{pmatrix}
\]
Here, $I$ is the $k_i\times k_i$ identity matrix, $J$ is the $k_i\times m_i$ matrix with all entries equal to $1$, and all empty blocks are 0.
Take any $\ve{u}\in\NN^{m_{i+1}}$ and $\ve{v}\in\{0,1,\dots,d-1\}^{m_{i+1}}$, and write
\begin{align*}
\ve{u}^\intercal = (\vx{1}^\intercal \,\mid \dots \mid \vx{d}^\intercal\, \mid {\ve{z}}^\intercal\,) & \textrm{ for }\vx{1}\,,\dots,\vx{d}\, \in\NN^{m_{i}}, \ve{z} \in \NN^{k_i};\\
\ve{v}^\intercal=(\vxp{1}^\intercal \mid \dots \mid \vxp{d}^\intercal \mid {\ve{z}'}^\intercal) & \textrm{ for }\vxp{1},\dots,\vxp{d} \in\{0,\dots,d-1\}^{m_{i}}, \ve{z}' \in \{0,\dots,d-1\}^{k_i}.
\end{align*}
Then $M\ve{u} = M\ve{v}$ is equivalent to:
\begin{align}
\sum_{i=1}^d B \vx{i} + I \ve{z} &=\sum_{i=1}^d B \vxp{i} + I \ve{z}' \\
B \vx{1} + (J-B) \vx{i} &= B \vxp{1} + (J-B) \vxp{i}  &\quad\quad (i=2\dots d)\\
 \sum_{j=1}^{m_i} x^{(i)}_j &= \sum_{j=1}^{m_i} x'^{(i)}_j &\quad\quad (i=2\dots d)\\
 \sum_{j=1}^{k_i} z_j &= \sum_{j=1}^{k_i} z_j'
\vspace*{-10pt}\end{align}

Equation (3) is equivalent to $J \vx{i} = J \vxp{i}$ (for $i=2\dots d$), thus the sum of (1) with all (2) equations implies $d B \vx{1} + \ve{z} =  dB\vxp{1} + \ve{z}'$ and hence $z_j \equiv z_j' \pmod{d}$ for all $j$.
Since $z_j\geq 0$ and $z_j' \in \{0,\dots,d-1\}$, this implies $z_j\geq z_j'$.
This together with (4) implies that in fact $z_j =  z_j'$ for all $j$, and thus $\ve{z}=\ve{z}'$.

Since $I\ve{z} = I \ve{z}'$ and  $J\vx{i} = J\vxp{i}$ ($i=2\dots d$), linear combinations of equations (1) and (2) imply that  $B \vx{i} = B \vxp{i}$ for each $i=1\dots d$.
By inductive assumption on $B$, this implies that $\vx{i}=\vxp{i}$ and hence $\ve{u}=\ve{v}$.

\medskip
We thus obtain $\{0,1\}$-matrices $M_i$ with $k_i$ rows and $m_i$ columns such that
\begin{equation*}
k_1=m_1=d\quad \textrm{and}\quad k_{i+1}=dk_{i}+d\quad \textrm{and}\quad m_{i+1} = dm_{i}+k_{i}.
\end{equation*}
Using a straightforward induction we can check the following explicit formulas for $k_i$ and $m_i$:
\begin{equation*}
k_i = \frac{d^{i+1}-d}{d-1}\quad\textrm{and}\quad m_i = \frac{(i-1) \cdot d^i}{d-1} + \frac{C d^i+d}{(d-1)^2},\quad \textrm{where }C=(d-1)(d-2).
\end{equation*}
Hence $k_i \leq \frac{d \cdot d^{i}}{d-1}$, $m_i \geq \frac{(i-1) \cdot d^{i}}{d-1}$, and $\log_d(m_i)\leq(i-1)(1+o(1))$, implying $k_i \leq \frac{m_i \cdot d}{\log_d m_i} (1+o(1))$.
To interpolate between $m_i$ and $m_{i+1}$, one can join several of the constructed matrices into a block-diagonal matrix.
Formally, if $f(m)$ denotes the least $k$ such that a $k \times m$ matrix as above exists, then $f(m+m') \leq f(m)+f(m')$ and $f(m_i) \leq  \frac{m_i \cdot d}{\log_d m_i} (1+o(1))$.
Standard methods then allow us to show that $f(m) \leq \frac{m \cdot d}{\log_d m} (1+o(1))$, see e.g.~\cite[Theorem 2]{Cantor66}; see also~\cite{MarKha89} for a slightly more explicit construction for all $m$.
\end{proof}

We remark that the bounds in Theorem~\ref{thm:detectingBounded} and Lemma~\ref{lem:detectingRandom} were also shown to be tight.
Lemma~\ref{lem:detectingOne} gives matrices that are also $d$-detecting, in particular, hence the bound is tight for $d=2$ (and tight up to an $\Oh(d)$ factor in general).

Note also that we can relax the non-negativity constraint to requiring that $\ve{u}\in\ZZ^m$ is any integer with all entries lower bounded by $-\lfloor\frac{d}{2}\rfloor$
and $\ve{v}\in\{-\lfloor\frac{d}{2}\rfloor,\dots,\lfloor\frac{d}{2}\rfloor\}^m$.
This is because $M \ve{u}=M\ve{v}$ is equivalent to $M(\ve{u}+\ve{c})=M(\ve{v}+\ve{c})$ where $\ve{c}$ is the constant $\lfloor\frac{d}{2}\rfloor$ vector.
This allows to use the same detecting matrix for such pairs of vectors as well.
However, note that some lower bound on the coefficients of $\ve{u}$ is necessary, since even if we fix $\ve{v}=0$, the matrix $M$ has a non-trivial kernel,
giving many non-zero vectors $\ve{u} \in \ZZ^m$ satisfying $M\ve{u}=M\ve{v}$.

\subsection{Coefficient reduction}
In further constructions, we will need a way to reduce coefficients in a given {\textsc{ILP Feasibility}} instance with a nonnegative constraint matrix $A$ to $\{0,1\}$.
We now prove that this can be done in a standard way by replacing each constraint with $\Oh(\log \|A\|_\infty)$ constraints that check the original equality bit by bit.
Here and throughout this paper we use the convention that for a vector $\vex$, by $x_i$ we denote the $i$-th entry of $\vex$.


\begin{lemma}[Coefficient Reduction]\label{lem:deltaReduction}
Consider an instance $\{A \vex = \veb, \vex\geq 0\}$ of {\textsc{ILP Feasibility}}, where $\veb \in \NN^k$ and $A$ is a nonnegative integer matrix with $k$ rows and $\ell$ columns.
In polynomial time, this instance can be reduced to an equivalent instance $\{A' \vex = \veb', \vex\geq 0\}$ of {\textsc{ILP Feasibility}}
where $A'$ is a $\{0,1\}$-matrix with $k' = \Oh(k \log \|A\|_\infty)$ rows and $\ell' = \ell+\Oh(k \log \|A\|_\infty)$ columns, and $\veb' \in \NN^{k'}$ is a vector with $\|\veb'\|_\infty = \Oh(\|\veb\|_\infty)$.
\end{lemma}
\begin{proof}
Denote $\delta = \lceil \log (1+\|A\|_{\infty})\rceil = \Oh(\log \|A\|_\infty)$.
Consider a single constraint $\vea^\intercal \vex = b$, where $\vea \in \NN^\ell$ is a row of $A$ and $b\in\NN$ is an entry of $\veb$.
Let $a_i[j]$ be the $j$-th bit of $a_i$, the $i$-th entry of vector $\vea$; similarly for $b$.
By choice of $\delta$, $\|\vea\|_\infty \leq 2^\delta - 1$, so each entry of $\vea$ has up to $\delta$ binary digits.
Now, for $\vex\in \ZZ^n$, the constraint $\vea^\intercal\vex = b$ is equivalent to
\[
\sum_{j = 0}^{\delta-1} 2^j \left( \sum_{i = 1}^n a_i[j] \cdot x_i \right) = b\,.
\]
We rewrite this equation into $\delta$ equations, each responsible for verifying one bit.
For this, we introduce $\delta-1$ carry variables $y_0,y_1,\ldots,y_{\delta-2}$ and emulate the standard algorithm for adding binary numbers by writing equations
\[
y_{j-1} + \sum_{i = 1}^n a_i[j] \cdot x_i = b[j] + 2y_j \qquad\qquad\textrm{for } j = 0, \ldots, \delta - 1,
\]
where $y_{-1}$ and $y_{\delta-1}$ are replaced with $0$ and $b[\delta-1]$ is replaced with the number whose binary digits are (from the least significant): $b[\delta-1],b[\delta],b[\delta+1],\dots$ (we do this because $b$ may have more than $\delta$ digits).
To get rid of the variable $y_j$ on the right-hand side, we let $B=2^{\lceil \log b\rceil}$ and introduce two new variables $y_j', y_j''$ for each carry variable $y_j$, with constraints
\[
	y_j + y_j' = B \quad\textrm{and}\quad y_j + y_j'' = B\quad\textrm{for }j=0,\ldots,\delta-2,
\]
which is equivalent to $y_j'=y_j''=B - y_j$.
Hence the previous equations can be replaced by
\[
y_{j-1} + \sum_{i = 1}^n a_i[j] \cdot x_i + y_j'+y_j''= b[j] + 2B \qquad\qquad\textrm{for } j = 0, \ldots, \delta - 1.
\]
We thus replace each row of $A$ with $2(\delta-1)+\delta$ rows and $3(\delta-1)$ auxiliary variables.
\end{proof}

\subsection{Proof of Theorem~\ref{thm:main}}
The Exponential Time Hypothesis states that for some $c>0$, {\textsc{3SAT}} with $n$ variables cannot be solved in time $\Ohstar(2^{cn})$ (the $\Ohstar$ notation hides polynomial factors).
It was introduced by Impagliazzio, Paturi, and Zane~\cite{seth} and developed by Impagliazzo and Paturi~\cite{eth} to become a central conjecture for proving tight lower bounds for the complexity of various problems.
While the original statement considers the parameterization by the number of variables, the \emph{Sparsification Lemma}~\cite{eth} allows us to assume that the number of clauses is linear in the number of variables,
and hence we have the following.

\begin{theorem}[{see e.g.~\cite[Theorem 14.4]{platypus}}]
\label{thm:eth-main}
Unless ETH fails, there is no algorithm for {\textsc{3SAT}} that runs in time $2^{o(n+m)}$, where $n$ and $m$ denote the numbers of variables and clauses.
\end{theorem}

We now proceed to the proof of Theorem~\ref{thm:main}.
Our first step is to decrease the number of occurrences of each variable.
The {\textsc{(3,4)SAT}} is the variant of {\textsc{3SAT}} where each clause uses exactly 3 different variables and every variable occurs in at most 4~clauses.
Tovey~\cite{Tovey84} gave a linear reduction from {\textsc{3SAT}} to {\textsc{(3,4)SAT}}, i.e., an algorithm that, given an instance of {\textsc{3SAT}} with $n$ variables and $m$ clauses,
in linear time constructs an equivalent instance of {\textsc{(3,4)SAT}} with $\Oh(n+m)$ variables and clauses.
In combination with Theorem~\ref{thm:eth-main} this yields:

\begin{corollary}\label{cor:34sat}
Unless ETH fails, there is no algorithm for {\textsc{(3,4)SAT}} that runs in time $2^{o(n+m)}$, where $n$ and $m$ denote the numbers of variables and clauses, respectively.
\end{corollary}

We now reduce {\textsc{(3,4)SAT}} to {\textsc{ILP Feasibility}}.
A {\textsc{$(3,4)$SAT}} instance $\varphi$ with $n$ variables and $m$ clauses can be encoded in a standard way as an {\textsc{ILP Feasibility}} instance
with $\Oh(n+m)$ variables and constraints as follows.
For each formula variable $v$ we introduce two ILP variables $x_v$ and $x_{\neg v}$ with a constraint $x_v+x_{\neg v}=1$ (hence exactly one of them should be 1, the other 0).
For each clause $c$ we introduce two auxiliary slack variables $y_c,z_c$ and two constraints: $y_c+z_c=2$ and $x_{\ell_1}+x_{\ell_2}+x_{\ell_3}+y_c = 3$, where $\ell_1,\ell_2,\ell_3$ are the three literals in~$c$.
Since $y_c,z_c$ will not appear in any other constraints, the first constraint is equivalent to ensuring that $y_c \leq 2$, so the second constraint is equivalent to $x_{\ell_1}+x_{\ell_2}+x_{\ell_3} \geq 1$.
This way, one can reduce  in polynomial time a {\textsc{$(3,4)$SAT}} instance $\varphi$ with $n$ variables and $m$ clauses into an equivalent instance $\{\vex \in \ZZ^{\ell} \mid A\vex = \veb, \vex \geq 0\}$ of {\textsc{ILP feasibility}} where:
\begin{itemize}
\item the constraint matrix $A$ has $k:=n+2m$ rows and $\ell:=2n+2m$ columns;
\item each entry in $A$ is zero or one;
\item each row and column of $A$ contains at most $4$ non-zero entries; and
\item the target vector $\veb$ has all entries equal to $1$, $2$, or $3$; 
\end{itemize}

\newcommand\var{{\text{var}}}
\newcommand\cl{{\text{cl}}}
We now reduce the obtained instance to another {\textsc{ILP Feasibility}} instance containing only $\Oh((n+m) / \log (n+m))$ constraints.
Let $M$ be the detecting matrix given by Lemma~\ref{lem:detectingOne} for $d=4$ and the required number of columns ($m$ in the notation of the statement of Lemma~\ref{lem:detectingOne})
equal to the number or rows (constraints) in $A$, which is $k$.
Then for any $\vex \in \NN^{\ell}$, we have $A \vex \in \NN^{k}$ (since $A$ is non-negative) and $\veb \in \{0,\dots,d-1\}^{k}$,
hence by Lemma~\ref{lem:detectingOne} we have that $A\vex=\veb$ if and only if $MA\vex = M \veb$.
We conclude that the {\textsc{ILP Feasibility}} instance $\{\vex \in \ZZ^{\ell} \mid A' \vex = \veb', \vex \geq 0\}$ with $A'=MA$ and $\veb'=M\veb$ is equivalent
to the previous instance $\{\vex \in \ZZ^{\ell} \mid A\vex = \veb, \vex \geq 0\}$.

The new instance has the same number $\ell'=\ell=2n+2m$ of variables, but only $k'=\Oh(k /\log k)=\Oh((n+m) / \log (n+m))$ constraints.
The entries of $\veb'=M\veb$ are non-negative and bounded by $k \cdot \|\veb\|_\infty = \Oh(n+m)$.
Similarly, entries of $A'=MA$ are non-negative, and since every column of $A$ has at most 4 non-zero entries, we get $\|A'\|_\infty \leq 4$.

To further reduce $\|A'\|_\infty$, we apply Lemma~\ref{lem:deltaReduction}, replacing each row of $A'$ by a constant number of $\{0,1\}$-rows and auxiliary variables.
This way, we reduced in polynomial time a {\textsc{$(3,4)$SAT}} instance $\varphi$ with $n$ variables and $m$ clauses into an equivalent {\textsc{ILP Feasibility}} instance $\{\vex \in \ZZ^{\ell''}\mid A''\vex = \veb'',\vex\geq 0\}$,
where $A''$ is a $\{0,1\}$-matrix with $\ell'' = \ell'+\Oh(k')=\Oh(n+m)$ columns and $k''=\Theta(k')=\Theta((n+m) / \log (n+m))$ rows, while $\|\veb''\|_\infty = \Oh(n+m)$.
Hence $\ell'',\|\veb''\|_\infty = \Oh(k'' \log k'')$.

We are now in position to finish the proof of Theorem~\ref{thm:main}.
Suppose there is an algorithm for {\textsc{ILP Feasibility}} that works in time $2^{o(k'' \log k'')}$ on instances with $A \in \{0,1\}^{k''\times \ell''}$ and $\ell'',\|\veb''\|_\infty = \Oh(k'' \log k'')$.
Then applying the above reduction would solve {\textsc{(3,4)SAT}} instances with $N=n+m$ variables and clauses in time
$2^{o((N/\log N) \cdot \log (N/\log N))} = 2^{o(N)}$,
which contradicts ETH by Corollary~\ref{cor:34sat}.
This concludes the proof of Theorem~\ref{thm:main}.

\subsection{Reducing coefficients in the target vector}
We now prove Corollary~\ref{cor:main}. That is, we show that in Theorem~\ref{thm:main}, the coefficients in the target vector $\veb$ can be reduced to constant, at the cost of introducing negative (-1) coefficients in the matrix $A$.

\begin{proof}[Proof of Corollary~\ref{cor:main}]
	Let $\{A \vex = \veb, \vex\geq 0\}$ be an instance given by Theorem~\ref{thm:main}, with $\ell,\|\veb\|_\infty = \Oh(k \log k)$.
	Let $s := \lceil \log(\|\veb\|_\infty + 1)\rceil$.
	To the system of linear equalities we add $s+1$ new variables $z,y_0,\dots,y_{s-1}$,
	with constraints $z=1$ and
	\[
	z + y_0 + \cdots + y_{i-1} = y_i  \qquad\qquad\qquad  \forall i = 0, \ldots, s-1 \,,
	\]
	which force $z=1$ and $y_i = 2^i$.
	Then each original constraint $\vea^\intercal \vex = b$, where $\vea \in \NN^\ell$ is a row of $A$ and $b\in\NN$ is an entry of $\veb$,
	can be replaced by the constraint $\vea^\intercal \vex - \ve{c}^\intercal\ve{y} = 0$,
	where $\ve{c}$ is chosen so that $\ve{c}^\intercal \ve{y} = b$;
	that is, the $i$-th entry of $\ve{c}$ is 0 or 1 depending on the $i$-th bit of $b$.

	In matrix form, we thus created the following instance (where $B \in \{-1,0\}^{(1+s) \times \ell}$ is the matrix corresponding to the binary encoding of $\veb$, with the first column zero, since it corresponds to the variable $z$):
	\begin{center}
	$$
	\left(\begin{NiceArray}{C|C}[name=Ap]
		A\rule[-20pt]{0pt}{50pt} & B\\\hline
		\hspace*{5em} & \small\begin{matrix}1&&&\\1&-1&&\\1&1&-1&\\1&1&1&-1\end{matrix}\\
	\end{NiceArray}\right)
	=
	\left(\begin{NiceArray}{C}
		\small\begin{matrix}0\\0\\0\\\end{matrix}\rule[-20pt]{0pt}{50pt}\\\hline
		\small\begin{matrix}1\\0\\0\\0\\\end{matrix}\!\\
	\end{NiceArray}\right)
	$$
	\tikz[remember picture,overlay] \node at (Ap-1-1.north) {$\ell$};
	\tikz[remember picture,overlay] \node at (Ap-1-2.north |- Ap-1-1.north) {$1+s$};
	\tikz[remember picture,overlay] \node at ($(Ap-2-1.west |- Ap-1-1.west)-(0.8,0)$) {$k$};
	\tikz[remember picture,overlay] \node at ($(Ap-2-1.west)-(0.8,0)$) {$1+s$};
	\end{center}
	Since $s = \Oh(\log k)$, the resulting matrix has $\{-1,0,1\}$ entries, $k+1+s = \Theta(k)$ rows, $\ell+1+s=\Oh(k \log k)$ columns. The new target vector has only $\{0,1\}$ entries, as required.
\end{proof}

\section{Parameterization by the dual treedepth}

\subsection{Preliminaries}

\subparagraph*{Treedepth and dual treedepth.}
For a graph $G$, the {\em{treedepth}} of $G$, denoted $\td(G)$, can be defined recursively as follows:
\begin{equation}\label{eq:rec-td}
\td(G)=
\begin{cases}
1 & \quad \textrm{if $G$ has one vertex;}\\[0.2cm]
\max(\td(G_1),\ldots,\td(G_p)) & \quad \textrm{if $G$ is disconnected and $G_1,\ldots,G_p$} \\
                               & \quad \textrm{are its connected components;} \\[0.2cm]
1+\min_{u\in V(G)}\td(G-u) & \quad \textrm{if $G$ has more than one vertex}\\
                               & \quad \textrm{and is connected.}
\end{cases}
\end{equation}
See e.g.~\cite{treedepth}.
Equivalently, treedepth is the smallest possible height of a rooted forest $F$ on the same vertex set as $G$
such that whenever $uv$ is an edge in $G$, then $u$ is an ancestor of $v$ in $F$ or vice versa.

Since we focus on constraints, we consider, for a matrix $A$,
the {\em{constraint graph}} or {\em{dual graph}} $G_D(A)$, defined as the graph with rows of $A$ as vertices where two rows are adjacent if and only if in some column they simultaneously contain a non-zero entry.
The {\em{dual treedepth}} of $A$, denoted $\dualtd(A)$, is the treedepth of $G_D(A)$.

The recursive definition~\eqref{eq:rec-td} is elegantly reinterpreted in terms of row removals and partitioning into blocks as follows.
A matrix $A$ is {\em{block-decomposable}} if after permuting its rows and columns it can be presented in block-diagonal form, i.e., rows and columns can be partitioned into intervals $R_1,\ldots,R_p$ and $C_1,\ldots,C_p$, for some $p\geq 2$, such that non-zero entries appear only in blocks $B_1,\ldots,B_p$, where $B_i$ is the block of entries at intersections of rows from $R_i$ with columns from $C_i$.
It is easy to see that $A$ is block-decomposable if and only if $G_D(A)$ is disconnected, and the finest block decomposition of $A$ corresponds to the partition of $G_D(A)$ into connected components. The blocks $B_1,\ldots,B_p$ in this finest partition are called the {\em{block components}} of $A$---they are not block-decomposable.
Then the recursive definition of treedepth provided in~\eqref{eq:rec-td} translates to the following definition of the dual treedepth of a matrix~$A$:
\begin{equation}\label{eq:rec-dualtd}
\dualtd(A)=
\begin{cases}
1 & \quad \textrm{if $A$ has one row;}\\[0.16cm]
\max(\dualtd(B_1),\ldots,\dualtd(B_p)) & \quad \textrm{if $A$ is block-decomposable and}\vspace*{-3pt} \\
                               & \quad \textrm{$B_1,\ldots,B_p$ are its block components;} \\[0.16cm]
1+\min\limits_{\vea^\intercal \colon \textrm{rows of }A}\,\dualtd(A{\backslash \vea^\intercal})\vspace*{-6pt} & \quad \textrm{if $A$ has more than one row and}\\
			       & \quad \textrm{is not block decomposable.}
\end{cases}
\end{equation}
Here $A{\backslash\vea^\intercal}$ is the matrix obtained from $A$ by removing the row $\vea^\intercal$.
Intuitively, dual treedepth formalizes the idea that a block-decomposable matrix is as hard as the hardest of its block components, and that adding a single row makes it a bit harder, but not uncontrollably so.

\subparagraph*{Graver bases.}
Two integer vectors $\vea,\veb\in \ZZ^n$ are {\em{sign-compatible}} if $a_i\cdot b_i\geq 0$ for all $i=1,\ldots,n$.
For $\vea,\veb\in \ZZ^n$ we write $\vea\cfleq \veb$ if $\vea$ and $\veb$ are sign-compatible and $|a_i|\leq |b_i|$ for all $i=1,\ldots,n$.
Then $\cfleq$ is a partial order on $\ZZ^n$; we call it the {\em{conformal order}}. Note that $\cfleq$ has a unique minimum element, which is the zero vector $\ve{0}$.

For a matrix $A$, the {\em{Graver basis}} of $A$, denoted $\Graver(A)$ is the set of conformally minimal vectors in $(\ker A\cap \ZZ^n)\setminus\{\ve{0}\}$.
It is easy to see by Dickson's lemma that $(\ZZ^n,\cfleq)$ is a well quasi-ordering, hence there are no infinite antichains with respect to the conformal order.
It follows that the Graver basis of every matrix is finite, though it can be quite large.
For a matrix $A$ and $p\in [1,\infty]$, we denote $g_p(A)=\max_{\veu\in \Graver(A)} \|\veu\|_p$.

\subsection{Upper bound}
We start with the upper bound for the dual treedepth parameterization, that is, Theorem~\ref{thm:dualtd-upper-bound}.
As explained in the introduction, this result easily follows from the work of Kouteck\'y et al.~\cite{KouteckyLO18} and
the following lemma bounding $g_1(A)$ in terms of $\dualtd(A)$ and $\|A\|_\infty$, for any integer matrix $A$.

\begin{lemma}\label{lem:g1-dualtd}
For any matrix $A$ with integer entries, it holds that
$$g_1(A)\leq (2\|A\|_\infty+1)^{2^{\dualtd(A)}-1}.$$
\end{lemma}

Before we prove Lemma~\ref{lem:g1-dualtd}, let us sketch how using the reasoning from Kouteck\'y et al.~\cite{KouteckyLO18} one can derive Theorem~\ref{thm:dualtd-upper-bound}.
Using the bound on the $\ell_1$-norm of vectors in the Graver basis of~$A$, we can construct a {\em{$\Lambda$-Graver-best oracle}} for the considered {\textsc{ILP Optimization}} instance.
This is an oracle that given any feasible solution $\vex$, returns another feasible solution $\vex'$ that differs from $\vex$ only by an integer multiple not larger than $\Lambda$
of a vector from the Graver basis of~$A$,
and among such solution achieves the best goal value of $\vew^\intercal\vex'$. Such a $\Lambda$-Graver-best oracle runs in time $(\|A\|_\infty\cdot g_1(A))^{\Oh(\dualtw(A))}\cdot |I|^{\Oh(1)}$, where $\dualtw(A)$
is the treewidth of the constraint graph $G_D(A)$, which is always upper bounded by $\dualtd(A)+1$. See the proof of Lemma~25 and the beginning of the proof of Theorem~3 in~\cite{KouteckyLO18};
the reasoning there is explained in the context of tree-fold ILPs, but it uses only boundedness of the dual treedepth of~$A$. Once a $\Lambda$-Graver-best oracle is implemented,
we can use it to implement a {\em{Graver-best oracle}} (Lemma~14 in~\cite{KouteckyLO18}) within the same asymptotic running time, and finally use the main theorem---Theorem~1 in~\cite{KouteckyLO18}---to obtain the algorithm promised in Theorem~\ref{thm:dualtd-upper-bound}~above.

We now proceed to the proof of Lemma~\ref{lem:g1-dualtd}.

\begin{proof}[Proof of Lemma~\ref{lem:g1-dualtd}]
We proceed by induction on the number of rows of $A$ using the recursive definition~\eqref{eq:rec-dualtd}. For the base case---when $A$ has one row---we may use the following well-known bound.

\begin{claim}[Lemma 3.5.7 in~\cite{deLoeraBook}]\label{cl:base-case}
If $A$ is an integer matrix with one row, then $$g_1(A)\leq 2\|A\|_{\infty}+1.$$
\end{claim}
We note that the original bound of $2\|A\|_{\infty}-1$, stated in~\cite{deLoeraBook}, works only for non-zero $A$.

We now move to the induction step, so suppose the considered matrix $A$ has more than one row.
We consider two cases: either $A$ is block-decomposable, or it is not.

First suppose that $A$ is block-decomposable.
Let $B_1,\ldots,B_p$ be the block components of~$A$, and let $R_1,\ldots,R_p$ and $C_1,\ldots,C_p$ be the corresponding partitions of rows and columns of~$A$ into segments, respectively.
Observe that integer vectors $\veu$ from $\ker A$ are exactly vectors of the form $(\,\vev^{(1)}\ |\ \vev^{(2)}\ |\ \ldots\ |\ \vev^{(p)}\,)$,
where each $\vev^{(i)}$ is an integer vector of length $|C_i|$ that belongs to $\ker B_i$.
It follows that $\Graver(A)$ consists of vectors of the following form:
for some $i\in \{1,\ldots,p\}$ put a vector from $\Graver(B_i)$ on coordinates corresponding to the columns of $C_i$, and fill all the other entries with zeroes.
Consequently, we have
\begin{equation}\label{eq:decomposable1}
g_1(A)\leq \max_{i=1,\ldots,p} g_1(B_i).
\end{equation}
On the other hand, by~\eqref{eq:rec-dualtd} we have
\begin{equation}\label{eq:decomposable2}
\dualtd(A)=\max_{i=1,\ldots,p} \dualtd(B_i).
\end{equation}
Since each matrix $B_i$ has fewer rows than $A$, we may apply the induction assumption to matrices $B_1,\ldots,B_p$, thus inferring by~\eqref{eq:decomposable1} and~\eqref{eq:decomposable2} that
$$g_1(A)\leq \max_{i=1,\ldots,p} g_1(B_i)\leq \max_{i=1,\ldots,p} (2\|B_i\|_\infty+1)^{2^{\dualtd(B_i)}-1}\leq (2\|A\|_\infty+1)^{2^{\dualtd(A)}-1}.$$

We are left with the case when $A$ is not block-decomposable.
For this, we use the following claim, which is essentially Lemma 3.7.6 and Corollary~3.7.7 in~\cite{deLoeraBook}.
The statement there is slightly different, but the same proof, which we repeat for convenience in the appendix, in fact proves the following bound.

\begin{claim}\label{cl:step-non-decomposable}
Let $A$ be an integer matrix and let $\vea^\intercal$ be a row of $A$. Then
$$g_1(A)\leq (2\|\vea^\intercal\|_\infty+1)\cdot g_1(A\backslash\vea^\intercal)\cdot g_\infty(A\backslash\vea^\intercal).$$
\end{claim}
\begin{proof}
Denote $B=A\backslash\vea^\intercal$.
Consider any vector $\veu\in \Graver(A)$. Then $\veu$ is also in $\ker B\cap \ZZ^n$, where $n$ is the number of columns of $A$,
which means that we can write $\veu$ as a {\em{sign-compatible sum}} of elements of the Graver basis of $B$, that is,
$$\veu = \sum_{i=1}^p \lambda_i \veg_i,$$
for some $\lambda_1,\ldots,\lambda_p\in \NN$ and distinct sign-compatible vectors $\veg_1,\ldots,\veg_p\in \Graver(B)$.

Let $\velambda$ be a vector of length $p$ with entries $\lambda_1,\ldots,\lambda_p$.
Further, let $\veb$ be also a vector of length $p$, where $b_i=\vea^\intercal \veg_i$.
Considering $\veb^\intercal$ as a matrix with one row, we have $\velambda\in \ker \veb^\intercal$.
Indeed, we have
$$\sum_{i=1}^p \lambda_i b_i = \sum_{i=1}^p \lambda_i (\vea^\intercal \veg_i)=\vea^\intercal \sum_{i=1}^p \lambda_i \veg_i = \vea^\intercal \veu = 0,$$
because $\veu\in \ker \vea^\intercal$ due to $\veu\in \Graver(A)$.

We now verify that in fact $\velambda\in \Graver(\veb^\intercal)$.
Indeed, since $\veu$ is non-zero, $\velambda$ is non-zero as well.
Also, if there existed some non-zero $\velambda'\cflt \velambda$ with $\velambda'\in \ker \veb^\intercal$,
then the same computation as above would yield that $\veu'=\sum_{i=1}^p \lambda_i'\veg_i$ also belongs to $\ker A$.
However, as vectors $\veg_1,\ldots,\veg_p$ are sign-compatible, $\ve{0}\cflt \velambda'\cflt \velambda$ would entail $\ve 0\cflt \veu'\cflt \veu$, a contradiction with the conformal minimality of $\veu$
following from $\veu\in \Graver(A)$.

Now that we know that $\velambda\in \Graver(\veb^\intercal)$, we may use Claim~\ref{cl:base-case} to infer that
$$\|\velambda\|_1\leq 2\|\veb^\intercal\|_\infty+1=2\max_{i=1,\ldots,p} |\vea^\intercal \veg_i| +1 \leq 2\|\vea\|_\infty\cdot g_1(B)+1\leq (2\|\vea\|_\infty+1)\cdot g_1(B).$$
Hence, we have
$$\|\veu\|_1 = \left\|\sum_{i=1}^p \lambda_i \veg_i\right\|_1\leq \|\velambda\|_1\cdot g_{\infty}(B)\leq (2\|\vea\|_\infty+1)\cdot g_1(B)\cdot g_{\infty}(B).$$
This concludes the proof.
\cqed\end{proof}

Suppose then that $A$ is not block-decomposable.
By~\eqref{eq:rec-dualtd}, there exists a row $\vea^\intercal$ of $A$ such that $\dualtd(A\backslash\vea^\intercal)=\dualtd(A)-1$. Then, by Claim~\ref{cl:step-non-decomposable} and the inductive assumption, we have
\begin{align*}
g_1(A) & \leq (2\|\vea^\intercal\|_\infty+1)\cdot g_1(A\backslash\vea^\intercal)\cdot g_{\infty}(A\backslash\vea^\intercal)\leq (2\|A\|_\infty+1)\cdot \left(g_1(A\backslash\vea^\intercal)\right)^2\\
       & \leq (2\|A\|_\infty+1)^{1+2\cdot (2^{\dualtd(A)-1}-1)}=(2\|A\|_\infty+1)^{2^{\dualtd(A)}-1}.
\end{align*}
This concludes the proof.
\end{proof}

\subsection{Lower bound}

We now move to the proof of the lower bound, Theorem~\ref{thm:dualtd-lower-bound}.
We will reduce from the {\textsc{Subset Sum}} problem: given non-negative integers $s_1,\ldots,s_k,t$, encoded in binary, decide whether there is a subset of numbers $s_1,\ldots,s_k$ that sums up to~$t$.
The standard {\textsc{NP}}-hardness reduction from {\textsc{3SAT}} to {\textsc{Subset Sum}} takes an instance of {\textsc{3SAT}} with $n$ variables and $m$ clauses, and produces an instance $(s_1,\ldots,s_k,t)$ of
{\textsc{Subset Sum}} with a linear number of numbers and each of them of linear bit-length, that is, $k\leq \Oh(n+m)$ and $0\leq s_1,\ldots,s_k,t<2^\delta$, for some $\delta\leq \Oh(n+m)$.
See e.g.~\cite{AbboudBHS17} for an even finer reduction, yielding lower bounds for {\textsc{Subset Sum}} under Strong ETH.
By Theorem~\ref{thm:eth-main}, this immediately implies an ETH-based lower bound for {\textsc{Subset Sum}}.

\begin{lemma}\label{lem:SS-ETH}
Unless ETH fails, there is no algorithm for {\textsc{Subset Sum}} that would solve any input instance $(s_1,\ldots,s_k,t)$ in time $2^{o(k+\delta)}$,
where $\delta$ is the smallest integer such that $s_1,\ldots,s_k,t<2^\delta$.
\end{lemma}

The idea for our reduction from {\textsc{Subset Sum}} to {\textsc{ILP Feasibility}} is as follows. Given an instance $(s_1,\ldots,s_k,t)$,
we first {\em{construct}} numbers $s_1,\ldots,s_k$ using ILPs $P_1,\ldots,P_k$,
where each $P_i$ uses only constant-size coefficients and has dual treedepth $\Oh(\log \delta)$.
The ILP $P_i$ will have a designated variable $z_i$ and two feasible solutions: one that sets $z_i$ to $0$ and one that sets it to $s_i$.
Similarly we can construct an ILP $Q$ that forces a designated variable $w$ to be set to $t$.
Having that, the whole input instance can be encoded using one additional constraint: $z_1+\ldots+z_k-w=0$.
To construct each $P_i$, we first create $\delta$ variables $y_0,y_1,\ldots,y_{\delta-1}$ that are either all evaluated to $0$ or all evaluated to $2^0,2^1,\ldots,2^{\delta-1}$, respectively;
this involves constraints of the form $y_{j+1}=2y_j$. Then the number $s_i$ (or $0$) can be obtained on a new variable $z_i$ using a single constraint that assembles the binary encoding of $s_i$.
The crucial observations is that the constraint graph $G_D(P_i)$ consists of a path on $\delta$ vertices and one additional vertex, and thus has treedepth $\Oh(\log \delta)$.

We start implementing this plan formally by giving the construction for a single number~$s$.

\setcounter{MaxMatrixCols}{40}

\begin{lemma}\label{lem:construct-single}
For all positive integers $\delta$ and $s$ satisfying $0\leq s<2^\delta$, there exists an instance $P=\{A\vex = \veb,\, \vex\geq 0\}$ of {\textsc{ILP Feasibility}} with the following properties:
\begin{itemize}
\item $A$ has all entries in $\{-1,0,1,2\}$ and $\dualtd(A)\leq \log \delta + \Oh(1)$;
\item $\veb$ is a vector with all entries in $\{0,1\}$; and
\item $P$ has exactly two solutions $\vex^{(1)}$ and $\vex^{(2)}$, where $x^{(1)}_1=0$ and $x^{(2)}_1=s$.
\end{itemize}
Moreover, the instance $P$ can be constructed in time polynomial in $\delta+\log s$.
\end{lemma}
\begin{proof}
We shall use $n+2$ variables, denoted for convenience by $y_0,y_1,\ldots,y_{\delta-1},z,u$; these are arranged into the variable vector $\vex$ of length $\delta+2$ so that $x_1=z$.
Letting $b_0,b_1,\ldots,b_{\delta-1}$ be the consecutive digits of the number $s$ in the binary encoding, the instance $P$ then looks as follows:
$$\begin{matrix}
u & + & y_0    &   &        &   &         &   &                          &   &                          &   &   & = & 1       \\
  &   & 2y_0   & - & y_1    &   &         &   &                          &   &                          &   &   & = & 0       \\
  &   &        &   & 2y_1   & - &  y_2    &   &                          &   &                          &   &   & = & 0       \\
  &   &        &   &        &   & \ddots  &   & \ddots                   &   &                          &   &   &   & \vdots  \\
  &   &        &   &        &   &         &   & 2y_{\delta-2}            & - & y_{\delta-1}             &   &   & = & 0       \\
  &   & b_0y_0 & + & b_1y_1 & + & \ldots  & + & b_{\delta-2}y_{\delta-2} & + & b_{\delta-1}y_{\delta-1} & - & z & = & 0
\end{matrix}$$
Since $0\leq u \leq 1$, it is easy to see that $P$ has exactly two solutions in nonnegative integers:
\begin{itemize}
\item If one sets $u=1$, then all the other variables need to be set to $0$.
\item If one sets $u=0$, then $y_i$ needs to be set to $2^i$ for all $i=0,1,\ldots,\delta-1$, and then $z$ needs to be set to $s$ by the last equation.
\end{itemize}
It remains to analyze the dual treedepth of $A$.
Observe that the constraint graph $G_D(A)$ consists of a path of length $\delta$, plus one vertex corresponding to the last equation that may have an arbitrary neighborhood within the path.
Since the path on $\delta$ vertices has treedepth $\lceil \log (\delta+1)\rceil$, it follows that $G_D(A)$ has treedepth at most $1+\lceil \log (\delta+1)\rceil\leq \log \delta + \Oh(1)$.
\end{proof}

We note that in the above construction one may remove the variable $u$ and replace the constraint $u+y_0=1$ with $y_0=1$, thus forcing only one solution: the one that sets the first variable to $s$.
This will be used later.

We are ready to show the core part of the reduction.

\begin{lemma}\label{lem:SS-ILPFeas}
An instance $(s_1,\ldots,s_k,t)$ of {\textsc{Subset Sum}} with $0\leq s_i,t<2^\delta$ for $i=1,\ldots,k$, can be reduced in polynomial time to an equivalent instance
$\{A\vex = \veb,\, \vex\geq 0\}$ of {\textsc{ILP Feasibility}} where entries of $A$ are in $\{-1,0,1,2\}$, entries of $\veb$ are in $\{0,1\}$, and $\dualtd(A)\leq \log \delta + \Oh(1)$.\looseness=-1
\end{lemma}
\begin{proof}
For each $i\in \{1,\ldots,k\}$, apply Lemma~\ref{lem:construct-single} to construct a suitable instance $P_i=\{A_i\vex = \veb_i,\, \vex\geq 0\}$ of {\textsc{ILP Feasibility}} for $s=s_i$.
Also, apply Lemma~\ref{lem:construct-single} to construct a suitable instance $Q=\{C\vex = \ved,\, \vex\geq 0\}$ of {\textsc{ILP Feasibility}} for $s=t$, and
modify it as explained after the lemma's proof so that there is only one solution, setting the first variable to $t$.
Let
$$A=\begin{pmatrix} & & \vecc^\intercal & & \\ \hline A_1 & & & & \\ & A_2 & & & \\ & & \ddots & & \\ & & & A_k & \\ & & & & C \end{pmatrix}$$
where
$$\vecc^\intercal=(\, 1\, 0\, \ldots\, 0\ |\ 1\, 0\, \ldots\, 0\ |\ \ldots\ |\ 1\, 0\, \ldots\, 0\ |\ (-1)\, 0\, \ldots\, 0\, )$$
with consecutive blocks of lengths equal to the numbers of columns of $A_1,\ldots,A_k$, and $C$, respectively.
Observe that
$$\dualtd(A)\leq 1+\max(\dualtd(A_1),\ldots,\dualtd(A_k),\dualtd(C))=\log \delta + \Oh(1).$$
Further, let
$$\veb^\intercal=(\, 0\ |\ \veb^\intercal_1\ |\ \ldots\ |\ \veb^\intercal_k\ |\ \ved^\intercal\,).$$
We now claim that the ILP $\{A\vex =\veb,\, \vex\geq 0\}$ is feasible if and only if the input instance of {\textsc{Subset Sum}} has a solution.
Indeed, if we denote by $z_1,\ldots,z_k,w$ the variables corresponding to the first columns of blocks $A_1,\ldots,A_k,C$, respectively,
then by Lemma~\ref{lem:construct-single} within each block $A_i$ there are two ways of evaluating variables corresponding to columns of $A_i$: one setting $z_i=0$ and second setting $z_i=s_i$.
However, there is only one way of evaluating the variables corresponding to columns of $C$, which sets $w=t$.
The first row of $A$ then constitutes the constraint $z_1+\ldots+z_k-w=0$, which can be satisfied by setting $z_i$-s and $w$ as above if and only if some subset of the numbers $s_1,\ldots,s_k$ sums up to $t$.
\end{proof}

It remains to reduce entries in $A$ equal to 2, simply by duplicating variables.
\begin{lemma}\label{lem:stupid-deg-red}
An instance $\{A\vex = \veb,\, \vex\geq 0\}$ of {\textsc{ILP Feasibility}} where entries of $A$ are in $\{-1,0,1,2\}$ and entries of $\veb$ are in $\{0,1\}$ can be reduced in polynomial time to an equivalent instance
$\{A'\vex = \veb',\, \vex\geq 0\}$ of {\textsc{ILP Feasibility}} with all entries in $\{-1,0,1\}$ and $\dualtd(A')\leq \dualtd(A)+1$.
\end{lemma}
\begin{proof}
It suffices to duplicate each variable $x$ by introducing a variable $x'$, adding a constraint $x-x'=0$, and replacing occurrences of $2x$ in constraints by $x+x'$.
In the dual graph, this results in introducing a new vertex (for the constraint $x-x'=0$), adjacent only to those constraints that contained $x$, which form a clique in $G_D(A)$ (they are pairwise adjacent).
The new vertices are non-adjacent to each other.
We show that in total, this operation can only increase the dual treedepth by at most $1$.

Let $F$ be a rooted forest of height $\td(G_D(A))$ with the same vertex set as $G_D(A)$ such that whenever $uv$ is an edge of $G_D(A)$, then $u$ is an ancestor of $v$ in $F$ or vice versa.
Then in particular, for each original variable $x$, the constraints containing it form a clique in $G_D(A)$, so the constraint that is the lowest in $F$, say $\vea^\intercal$, has all the others as ancestors.
This means that the new vertex representing the constraint $x-x'=0$ can be added to $F$ as a pending leaf below $\vea^\intercal$.
Doing this for each original variable $x$ can only add pendant leaves to original vertices of $F$, which increases its height by at most $1$.
\end{proof}

Theorem~\ref{thm:dualtd-lower-bound} now follows by observing that combining the reductions of Lemma~\ref{lem:SS-ILPFeas} and Lemma~\ref{lem:stupid-deg-red} with a hypothetical algorithm for {\textsc{ILP Feasibility}}
on $\{-1,0,1\}$-input with running time $2^{2^{o(\dualtd(A))}}\cdot |I|^{\Oh(1)}$, or just $2^{o(2^{\dualtd(A)})}\cdot |I|^{\Oh(1)}$, would yield an algorithm for {\textsc{Subset Sum}} with running time $2^{o(k+\delta)}$, contradicting ETH by Lemma~\ref{lem:SS-ETH}.

\section{Conclusions}

We conclude this work by stating two concrete open problems in the topic.

First, apart from considering the standard form $\{A\vex =\veb,\, \vex\geq 0\}$,
Eisenbrand and Weismantel~\cite{EisenbrandW18} also studied the more general setting of ILPs of the form $\{A\vex =\veb,\, \vel\leq \vex\leq \veu\}$,
where $\vel$ and $\veu$ are integer vectors.
That is, instead of only requiring that every variable is nonnegative, we put an arbitrary lower and upper bound on the values it can take.
Note that such lower and upper bounds can be easily emulated in the standard formulation using slack variables, but this would require adding more constraints to the matrix $A$;
the key here is that we {\em{do not}} count these lower and upper bounds in the total number of constraints $k$.
For this more general setting, Eisenbrand and Weismantel~\cite{EisenbrandW18} gave an algorithm with running time $k^{\Oh(k^2)}\cdot \|A\|_{\infty}^{\Oh(k^2)}\cdot |I|^{\Oh(1)}$,
which boils down to $2^{\Oh(k^2 \log k)}\cdot |I|^{\Oh(1)}$ when $\|A\|_{\infty}=\Oh(1)$.
(A typo leading to a $2^{\Oh(k^2)}\cdot |I|^{\Oh(1)}$ bound has been fixed in later versions of the paper).
Is this running time optimal or could the $2^{\Oh(k^2 \log k)}$ factor be improved?
Note that Theorem~\ref{thm:main} implies a $2^{o(k\log k)}$-time lower-bound, unless ETH fails.

Second, in this work we studied the parameter dual treedepth of the constraint matrix $A$, but of course one can also consider the {\em{primal treedepth}}.
It can be defined as the treedepth of the graph over the columns (variables) of $A$, where two columns are adjacent if they have a non-zero entry in same row (the variables appear simultaneously in some constraint).
It is known that {\textsc{ILP Feasibility}} and {\textsc{ILP Optimization}} are fixed-parameter tractable when parameterized by $\|A\|_\infty$ and $\primtd(A)$,
that this, there is an algorithm with running time $f(\|A\|_\infty,\primtd(A))\cdot |I|^{\Oh(1)}$, for some function $f$~\cite{KouteckyLO18} (see also~\cite{EisenbrandHKKLO19}).
Again, the key ingredient here is an inequality on $\ell_\infty$-norms of the elements of the Graver basis of any integer matrix $A$: $g_\infty(A)\leq h(\|A\|_\infty,\primtd(A))$ for some function $h$.
The first bound on $g_\infty(A)$ was given by Aschenbrenner and Hemmecke~\cite{AschenbrennerH07}.
The work of Aschenbrenner and Hemmecke~\cite{AschenbrennerH07} considers the setting of {\em{multi-stage stochastic programming}} (MSSP),
which is related to primal treedepth in the same way as tree-fold ILPs are related to dual treedepth.
The translation between MSSP and primal treedepth was first formulated by Kouteck\'y et al.~\cite{KouteckyLO18}.
However, to establish a bound on $g_\infty(A)$, Aschenbrenner and Hemmecke~\cite{AschenbrennerH07} use the theory of well quasi-orderings (in a highly non-trivial way) and consequently give no direct bounds on the function $h$.
Recently, Klein~\cite{klein19-two-stage} gave the first constructive bound on $g_\infty(A)$ for MSSP.
However, we conjecture that the function $h$ has to be non-elementary in $\primtd(A)$.
If this was the case, an example could likely be used to prove a non-elementary lower bound under ETH for {\textsc{ILP Feasibility}} under that $\primtd(A)$ parameterization (with $\|A\|_\infty=\Oh(1)$).

Very recently, Kouteck{\'{y}} and Kr{\'{a}}{l'}~\cite{KouteckyK19} showed that algorithms parameterized by dual treedepth can be extended to the parameter ``branch-depth'' defined on the column matroid of the constraint matrix.
This parameter has the advantage of being invariant under row operations.
The transformation however incurs an exponential blow-up in the parameter.

\vfill
\pagebreak

\bibliography{ilp}

\end{document}